\documentclass{article} 
\usepackage[preprint]{neurips_2026}

\usepackage{amsmath,amsfonts,bm}

\def\eqref#1{equation~\ref{#1}}

\def\1{\bm{1}}

\DeclareMathAlphabet{\mathsfit}{\encodingdefault}{\sfdefault}{m}{sl}
\SetMathAlphabet{\mathsfit}{bold}{\encodingdefault}{\sfdefault}{bx}{n}

\usepackage[utf8]{inputenc} 
\usepackage[T1]{fontenc}    
\usepackage{hyperref}       
\usepackage{url}            
\usepackage{booktabs}       
\usepackage{amsfonts}       
\usepackage{nicefrac}       
\usepackage{microtype}      
\usepackage{xcolor}         

\usepackage{amsmath}
\usepackage{amsthm}
\usepackage{algorithm}
\usepackage{algorithmic}
\usepackage{amssymb}
\usepackage{graphicx}
\usepackage{comment}
\usepackage{multicol}
\usepackage{multirow}
\usepackage{colortbl}
\usepackage{xcolor}
\usepackage{enumitem}
\usepackage{newtxtext}

\newtheorem{theorem}{Theorem}

\newtheorem{corollary}{Corollary}

\title{PrivCert: Certifying Statement Support under Differential Privacy}

\author{Tsubasa Takahashi\thanks{\texttt{tsubasa.takahashi@acompany-ac.com}}\qquad \qquad Takumi Hiraoka\\
\\
Acompany Co., Ltd.\\
}

\begin{document}

\maketitle

\begin{abstract}
Differentially private (DP) text generation can protect individual records, but privacy alone does not specify what evidence a released statement carries about the underlying data.
We identify this as an \emph{evidence gap}: a private report may contain plausible claims without indicating whether they are strongly supported by the private dataset.
We introduce \emph{\textsc{PrivCert}}, a framework for privacy-preserving reporting that makes statement support explicit through privacy-preserving certificates and emit-or-abstain decisions.
As a canonical instantiation, \textsc{PrivCert}-PF (Proposal-and-Filter) separates data-independent candidate discovery from private support certification, emitting only statements whose support passes a private evidence test.
We provide theoretical grounding for this framework by characterizing the limits of implicit evidence under DP, deriving a sharp privacy--honesty frontier for single-statement certification, and establishing a worst-case cost for fine-grained multi-statement certification.
Experiments on synthetic tasks and TAB, WildChat, and Yelp show that explicit certification maintains low unsupported emission, while free-text DP baselines frequently produce low-support claims under the same declared support semantics.
We further show that the \textsc{PrivCert} contract can be realized with histogram, sparse-vector, and Gaussian mechanisms, and use DP synthetic data to illustrate an important boundary: support in a private proxy does not automatically certify support in the original data.
Together, these results position privacy-preserving reporting as an evidence-design problem: not only how to generate private text, but what a private report can substantiate about its underlying data.
\end{abstract}

\section{Introduction}
\label{sec:introduction}

Large language models (LLMs) are increasingly used to produce reports and insights from sensitive data, including retrieval-augmented systems over confidential corpora~\citep{lewis2020retrival}, enterprise assistants grounded in private organizational data~\citep{microsoft2026copilotprivacy}, and analytics pipelines over collections of user records~\citep{tamkin2025clio,
liu2025urania,cheu2025toward}.
Some such systems additionally run under confidential computing, where intermediate computations and logs may be inaccessible to deployers or auditors~\citep{apple2024pcc,ccc2023technical,google2025privateaiTechnical}.
This can create an \emph{observability blackout}, in which sensitive computation is protected while intermediate traces needed for inspection or auditing remain hidden.
In these settings, downstream readers may observe only the released report, making it important to know not only what it says, but how strongly its claims are supported by the underlying data.

Differential privacy (DP) provides a principled way to limit the influence of individual records on released outputs~\citep{dwork2006calibrating, dwork2014algorithmic}.
Recent work extends DP to natural-language generation through private training, synthetic-data generation, and inference-time private decoding~\citep{abadi2016deep,vinod2025invisibleink,xie2024differentially, yue2023synthetic}.
These methods protect individual records, but privacy alone does not specify what evidence a generated statement carries about the dataset.

\begin{figure}[t]
    \centering
    \includegraphics[width=.985\linewidth]{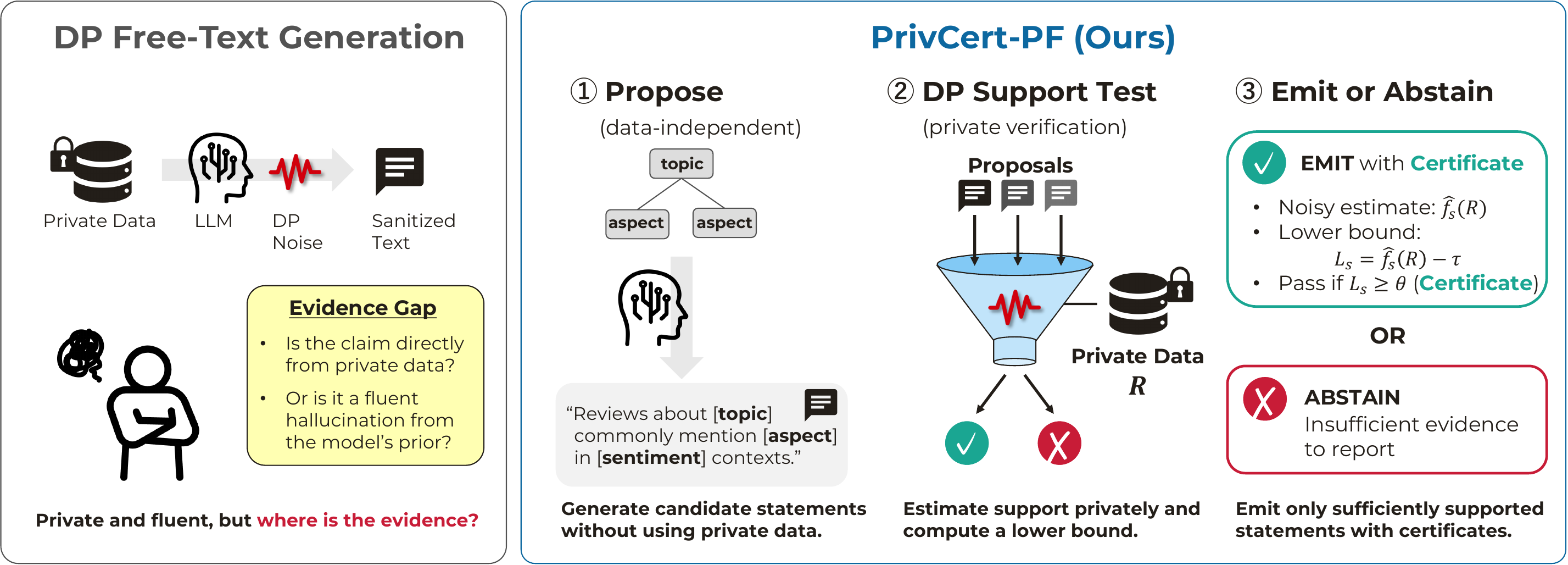}
    \caption{\textbf{From private text to private reporting with explicit support evidence.} DP free-text generation can produce fluent private text without exposing the evidential basis of its claims.
    \textsc{PrivCert}-PF (Proposal-and-Filter), one instantiation of \textsc{PrivCert}, separates candidate discovery from DP support verification and emits only statements that pass a privacy-preserving evidence test.}
    \label{fig:proposal-filter}
\end{figure}

A DP mechanism may output a plausible statement without indicating whether it reflects a widespread pattern in the private dataset, a rare observation amplified by randomness, or information inherited from the model.
For example, from ``Users frequently discuss medication side effects,'' a reader cannot tell how strongly the private corpus supports the claim.
We call this mismatch between private text and explicit support evidence the \emph{evidence gap}.

We formulate privacy-preserving reporting with explicit evidence and introduce \textsc{PrivCert}, a framework in which released statements carry privacy-preserving support evidence and are subject to emit-or-abstain decisions.
For a statement \(s\), let \(f_s(R)\in[0,1]\) denote its support in the private dataset \(R\).
\textsc{PrivCert} separates candidate discovery, private evidence construction, and release, without prescribing a particular candidate generator or DP primitive.

To ground this framework, we characterize limits on inferring support from ordinary DP outputs and derive a sharp privacy--honesty frontier for single-statement certification.
We instantiate the framework with \textsc{PrivCert}-PF (Proposal-and-Filter), which combines data-independent proposals with private support certification, and show that its sharply calibrated half-emission threshold is near-optimal.
We further characterize the worst-case cost of fine-grained multi-statement certification and evaluate alternative realizations of the same reporting contract.

\textbf{Contributions.}
We make four main contributions:
\begin{enumerate}[nosep, leftmargin=0.5cm]

\item We formulate the \emph{evidence gap} in privacy-preserving text generation and introduce \textsc{PrivCert}, a reporting framework that makes statement support explicit through privacy-preserving certificates and emit-or-abstain decisions (Section~\ref{sec:framework}).

\item We characterize fundamental limits on statement certification under DP: ordinary DP outputs provide only locally limited implicit evidence about support, and explicit certification is constrained by a sharp privacy--honesty frontier (Section~\ref{sec:limits}).

\item We introduce \textsc{PrivCert}-PF, a canonical proposal-and-filter instantiation that separates data-independent candidate discovery from private support certification, and characterize its honesty, coverage, and multi-statement scaling.
Under sharp calibration, its single-statement half-emission threshold is within one supporting record of optimal (Section~\ref{sec:method}).

\item Experiments on synthetic tasks and TAB, WildChat, and Yelp empirically characterize the evidence gap and explicit certification behavior, compare alternative \textsc{PrivCert} realizations, and identify a boundary between direct and proxy-data certification (Section~\ref{sec:experiments}).

\end{enumerate}

Together, these results shift privacy-preserving reporting from merely producing private text toward making its evidential basis explicit.
They reframe the central question from \emph{how to generate private text} to \emph{what a private report can certify about the data, and at what privacy cost}.
\section{Related Work}
\label{sec:related}

\textbf{Differentially private text generation.}
DP text generation has been studied through private training \citep{abadi2016deep}, inference-time decoding \citep{vinod2025invisibleink,thareja2025dp}, synthetic text generation \citep{xie2024differentially,yue2023synthetic}, and structured-text evaluation~\citep{wangstruct}.
These works focus on generating useful private text; we instead ask what evidence a released statement provides that the private data support it.

\textbf{Private aggregate reporting.}
CLIO summarizes sufficiently large conversation clusters \citep{tamkin2025clio}, while Urania provides formal DP guarantees through private clustering, partition selection, and histogram release \citep{liu2025urania}.
In contrast, we treat candidate statements themselves as objects whose support is privately verified against the dataset.

\textbf{Faithfulness and abstention.}
Faithfulness work studies whether generated statements are supported by source material~\citep{llm_hallucination,faithfulness_summarization}, while selective prediction and conformal methods use abstention under uncertainty \citep{geifman2017selective,conformal_prediction}.
Here, the source records are private, so abstention is driven by a privacy-preserving support test rather than direct verification.

\textbf{Private query answering, thresholding, and optimal mechanisms.}
Statement support \(f_s(R)\) is a statistical query, connecting our verifier to classical DP query answering~\citep{blum2008learning,hardt2010geometry,smith2011privacy,hay2010boosting}, sparse vector techniques~\citep{dwork2014algorithmic,lyu2017svt}, optimal private hypothesis testing~\citep{awan2018differentially}, and staircase mechanisms~\citep{geng2016optimal}.
Our focus is the reporting-specific interface between statements, support certificates, and abstention; our sharp frontier is induced by the dataset-wise one-sided honesty contract as empirical support crosses the reporting threshold.
\section{PrivCert: Private Reporting with Explicit Support Evidence}
\label{sec:framework}

We study privacy-preserving reporting about a private dataset \(R=\{r_1,\dots,r_n\}\) under a reporting task \(q\).
Our goal is to formalize not only whether a released report is differentially private, but also what evidence its statements carry about the underlying data.

A mechanism \(\mathcal A\) is \(\varepsilon\)-DP if, for each fixed \(q\),
\[
\Pr[\mathcal A(q,R)\in E]
\le
e^\varepsilon
\Pr[\mathcal A(q,R')\in E]
\]
for all adjacent datasets \(R\sim R'\).
By group privacy, datasets at replacement distance at most \(k\) satisfy the same inequality with \(e^{k\varepsilon}\).
This property underlies the lower bounds in Section~\ref{sec:limits}.

\textbf{Statements and evidence semantics.}
A statement \(s\in\mathcal S\) is a semantic claim about the dataset.
Its support is represented by \(f_s(R)\in[0,1]\).
For a fixed support predicate \(\sigma_s(r)\in\{0,1\}\), we use
\begin{equation}
\label{eq:support_classifier}
f_s(R)
=
\frac{1}{|R|}
\sum_{r\in R}\sigma_s(r).
\end{equation}
This empirical-frequency support has replace-one sensitivity at most
\(1/|R|\) when \(\sigma_s\) is fixed independently of \(R\).
More general semantic matchers may be used, provided their sensitivity and any private adaptation are included in the privacy accounting.

A reporting threshold \(\theta\in(0,1)\) defines insufficient support:
\[
f_s(R)<\theta
\quad\Longrightarrow\quad
s\text{ is unsupported}.
\]
Thus \(f_s\) specifies the operational notion of evidence, while \(\theta\) specifies how much evidence is sufficient for reporting.
The guarantees below are relative to this specified support semantics; DP does not itself establish matcher-independent semantic truth.

\textbf{\textsc{PrivCert} interface.}
A conventional free-text mechanism releases only
\(Y\sim\mathcal A(q,R)\), without explicitly stating how strongly the private dataset supports the claims appearing in \(Y\).
A \textsc{PrivCert} mechanism associates each evaluated candidate with
\(
(s_j,c_j,a_j),
\)
where \(s_j\) is a candidate statement, \(c_j\) is a privacy-preserving support certificate, and \(a_j\in\{\mathrm{emit},\mathrm{abstain}\}\) is the reporting decision.
The released report may expose emitted statements and their certificates, while abstention records that the available evidence was insufficient for release.
The interface separates candidate discovery, private evidence construction, and the final release decision; it does not prescribe a particular mechanism for any of these roles.

\textbf{Framework scope.}
\textsc{PrivCert} is a reporting contract rather than a new DP primitive.
Candidate discovery and evidence construction may use different mechanisms; any private access to \(R\) must be included in the privacy accounting, and the resulting certificate must satisfy the stated validity condition.
We state the basic results for pure DP for clarity, while the framework also allows approximate-DP mechanisms.
\textsc{PrivCert}-PF in Section~\ref{sec:method} is one direct-certification instantiation of this contract.

For the one-sided certificates studied here, \(c_s\) is valid at failure probability \(\beta\) if
\[
\Pr[c_s>f_s(R)]\le\beta.
\]
Equivalently, except with probability \(\beta\), the true dataset support is at least as large as the reported certificate \(c_s\).
Under the threshold reporting rule used by the instantiations studied here,
\(
a_s=\mathrm{emit}
\Longrightarrow
c_s\ge\theta,
\)
validity implies
\begin{equation}
\label{eq:framework-honesty}
f_s(R)<\theta
\quad\Longrightarrow\quad
\Pr[a_s=\mathrm{emit}]
\le\beta.
\end{equation}
This is the basic evidential contract (\textsc{PrivCert} contract): an unsupported statement is emitted only with controlled probability.
For multi-statement reports, the evidential error scope must also be
specified: per-candidate validity does not by itself control the probability that any unsupported statement appears in the complete report.

\section{Privacy Limits and Optimal Statement Certification}
\label{sec:limits}

We study two limits on statement evidence under differential privacy.
First, we ask how reliably support can be inferred from a single DP output.
Second, under the explicit one-sided honesty contract of
Section~\ref{sec:framework}, we characterize the largest possible emission
probability of any private certifier.
The latter yields a sharp frontier used to establish the near-optimality of
\textsc{PrivCert}-PF in Section~\ref{sec:method}.
Proofs are given in Appendix~\ref{app:proofs-limits}.

\subsection{Limits of Implicit Evidence}
\label{subsec:implicit-limits}

Fix a statement \(s\) and datasets \(R_0,R_1\) with
\(d(R_0,R_1)\le k\).
Let
\(
P_i:=\mathcal L(\mathcal A(q,R_i)),
\,
i\in\{0,1\}.
\)

\begin{theorem}[Local indistinguishability under DP]
\label{thm:testing-lb}
For any \(\varepsilon\)-DP mechanism \(\mathcal A\),
\[
d_{\mathrm{TV}}(P_0,P_1)
\le
\frac{e^{k\varepsilon}-1}{e^{k\varepsilon}+1}.
\]
Consequently, every possibly randomized binary test \(\varphi\) satisfies
\[
\mathbb E_{P_0}[\varphi(Y)]
+
\mathbb E_{P_1}[1-\varphi(Y)]
\ge
\frac{2}{e^{k\varepsilon}+1}.
\]
\end{theorem}

This is the standard binary-testing consequence of group privacy \citep{kairouz2015composition}, applied here to statement support.
In particular, the occurrence of a statement in a DP output cannot by itself serve as a calibrated support certificate.

The same argument limits support estimation.
Writing
\(a_i:=f_s(R_i)\) and \(\Delta:=|a_1-a_0|\), every estimator
\(\hat f_s:\mathcal Y\to[0,1]\) satisfies
\[
\max_{i\in\{0,1\}}
\mathbb E_{Y\sim P_i}
\bigl[
|\hat f_s(Y)-a_i|
\bigr]
\ge
\frac{\Delta}{2(e^{k\varepsilon}+1)}.
\]
For the empirical-frequency support in
Equation~\ref{eq:support_classifier}, changing support by \(\Delta\) requires at least \(\lceil n\Delta\rceil\) record replacements, so these bounds are local and become weak when \(k\varepsilon\) is large.
They do not imply that DP hides arbitrarily large support differences; rather, a single DP output cannot provide arbitrarily reliable implicit evidence on nearby datasets.

\subsection{Sharp Single-Statement Certification Frontier}
\label{subsec:sharp-frontier}

We now ask how quickly any private certifier can increase its probability of releasing a statement as its support rises above the reporting threshold.

For a fixed statement \(s\), let
\(
N_s(R):=\sum_{r\in R}\sigma_s(r),
\,
t_\theta:=\lceil n\theta\rceil,
\)
so that \(t_\theta-1\) is the largest unsupported support count.
Define
\[
d_s(R):=\max\{N_s(R)-t_\theta+1,0\},
\]
the number of supporting-record steps from this boundary.

\begin{theorem}[Sharp single-statement emission frontier]
\label{thm:sharp-frontier}
Assume \(0<\beta<1/2\) and that both supporting and non-supporting records exist.
Let
\(
\pi_0:=\beta,
\,
\pi_{k+1}
:=
\min\left\{
e^\varepsilon\pi_k,\,
1-e^{-\varepsilon}(1-\pi_k)
\right\}.
\)
For any \(\varepsilon\)-DP mechanism satisfying
\(
f_s(R)<\theta
\Longrightarrow
\Pr[a_s=\mathrm{emit}\mid R]\le\beta,
\)
the emission probability satisfies
\[
\Pr[a_s=\mathrm{emit}\mid R]
\le
\pi_{d_s(R)}
\]
for every dataset \(R\).
Moreover, the bound is pointwise tight: an
\(\varepsilon\)-DP, \(\beta\)-honest count-based mechanism attains
\(\pi_{d_s(R)}\) at every support count.
\end{theorem}

The theorem gives an exact upper envelope on how quickly release can become more likely as support increases.
Since \(\pi_k\le\beta e^{k\varepsilon}\), any certifier with emission probability at least \(1/2\) must satisfy
\[
d_s(R)
\ge
\left\lceil
\frac{1}{\varepsilon}
\log\frac{1}{2\beta}
\right\rceil.
\]
The optimal count-based mechanism attains the binary release frontier but provides only a threshold-level certificate.
Section~\ref{sec:method} shows that \textsc{PrivCert}-PF reaches this fundamental margin essentially optimally while additionally providing a quantitative lower certificate on statement support.
\section{A Canonical \textsc{PrivCert} Instantiation: Proposal-and-Filter}
\label{sec:method}

We instantiate \textsc{PrivCert} with \emph{\textsc{PrivCert}-PF (Proposal-and-Filter)}, which separates data-independent candidate discovery from private support certification and emits only statements whose lower support certificates clear the reporting threshold.
The construction is deliberately simple, making its privacy, evidential validity, and coverage behavior explicit.
Section~\ref{sec:experiments} evaluates alternative \textsc{PrivCert} instantiations with different sensitivity and selection structures.
Proofs and pseudocode are given in Appendix~\ref{app:proofs-method}.

\subsection{Proposal-and-Filter}
\label{sec:pf-algorithm}

Given a reporting task \(q\), let \(\Pi_q\) be a proposal distribution independent of the private dataset \(R\), and sample
\(
s_1,\dots,s_m \overset{\mathrm{i.i.d.}}{\sim}\Pi_q.
\)
For each candidate \(s_j\), \textsc{PrivCert}-PF computes
\[
\widehat f_{s_j}(R)=f_{s_j}(R)+Z_j,
\qquad
L_{s_j}=\widehat f_{s_j}(R)-\tau,
\]
and emits \(s_j\) iff \(L_{s_j}\ge\theta\).
Our primary instantiation uses independent Laplace noise with equal allocation of the total privacy budget; Section~\ref{sec:pf-guarantees} analyzes its calibration, and Appendix~\ref{app:support-model} gives the experimental accounting.

Candidate generation is therefore a discovery step rather than a source of evidential validity: proposals need not themselves be correct, but every released statement must pass the private support test.
Conversely, a supported statement that is never proposed cannot be certified.
Any final natural-language rendering is post-processing, but must preserve the semantics of the certified statement; otherwise the certificate applies only to the original candidate.

\subsection{Privacy, Honesty, and Near-Optimality}
\label{sec:pf-guarantees}

If support query \(j\) is \(\varepsilon_j\)-DP, then data-independent proposal, sequential composition, and post-processing make \textsc{PrivCert}-PF \(\sum_j\varepsilon_j\)-DP.

\begin{theorem}[Certificate validity and per-candidate one-sided honesty]
\label{thm:honesty}
Fix an evaluated statement \(s\).
Suppose
\(
\widehat f_s(R)=f_s(R)+Z,
\,
\Pr[Z\ge\tau]\le\beta,
\,
L_s=\widehat f_s(R)-\tau.
\)
Then
\[
\Pr[L_s>f_s(R)]\le\beta.
\]
Consequently, if \(s\) is emitted only when \(L_s\ge\theta\), then
\(
f_s(R)<\theta
\Longrightarrow
\Pr[\mathrm{Emit}_s]\le\beta.
\)
\end{theorem}

The guarantee remains valid after candidate selection when the certification noise is fresh and independent of the selection event; any private selection must additionally be included in the privacy accounting.
For \(m\) evaluated candidates, a union bound gives a report-level failure probability at most \(m\beta\); Appendix~\ref{app:family-wise} evaluates calibration to a fixed whole-report target.

\medskip
\noindent
\textbf{Near-optimality.}
\label{subsec:pf-optimality}
The sharp frontier in Theorem~\ref{thm:sharp-frontier} also characterizes how close \textsc{PrivCert}-PF is to the best possible single-statement certifier.
For a statement with privacy budget \(\varepsilon_s\), consider
\[
Z\sim\mathrm{Lap}\!\left(\frac{1}{n\varepsilon_s}\right),
\qquad
\tau=
\frac{1}{n\varepsilon_s}\log\frac{1}{2\beta},
\]
the tight symmetric Laplace calibration for the two one-sided tail conditions.

\begin{corollary}[Near-optimal sharply calibrated half-emission threshold]
\label{cor:pf-near-optimal}
Let
\(
A
:=
\frac{1}{\varepsilon_s}
\log\frac{1}{2\beta},
\,
k^\star:=\lceil A\rceil,
\)
and assume
\(
\left\lceil n\theta+A\right\rceil\le n,
\)
equivalently, \(\theta+\tau\le1\) for the sharp Laplace margin \(\tau=A/n\).

For any \(\varepsilon_s\)-DP, \(\beta\)-honest single-statement certifier that reaches emission probability at least \(1/2\), the first such distance from the largest unsupported dataset is at least \(k^\star\).
Sharply calibrated Laplace \textsc{PrivCert}-PF first reaches emission probability at least \(1/2\) at distance either \(k^\star\) or \(k^\star+1\).
Hence its half-emission threshold is within one supporting record of optimal.
\end{corollary}

Thus, under sharp calibration, little room remains to improve the half-emission threshold once the per-statement privacy budget is fixed.
The optimal staircase mechanism of Theorem~\ref{thm:sharp-frontier} achieves the exact binary release frontier but exposes only a threshold-level certificate; \textsc{PrivCert}-PF additionally provides a quantitative lower certificate on statement support.

\subsection{Discovery and Coverage}
\label{sec:pf-discovery}

Certification alone does not determine whether useful supported statements are discovered.
For \(a\in\mathbb{R}\), define
\[
p_a(R):=\Pr_{s\sim\Pi_q}[f_s(R)\ge a],
\]
which measures how much proposal mass lies above a given support level.

\begin{theorem}[EmitRate trade-off]
\label{thm:precision-coverage}
Let
\(s_1,\dots,s_m\overset{\mathrm{i.i.d.}}{\sim}\Pi_q\), 
and suppose
\(\hat f_{s_j}(R)=f_{s_j}(R)+Z_j\),
where \(Z_j\) is independent of \(s_j\) and satisfies
\(\Pr[Z_j\ge\tau]\le\beta\) and
\(\Pr[Z_j\le-\tau]\le\beta\).
Let
\[
O:=\{j:\widehat f_{s_j}(R)-\tau\ge\theta\},
\qquad
\mathrm{EmitRate}:=
\mathbb E\!\left[\frac{|O|}{m}\right].
\]
Then
\[
(1-\beta)p_{\theta+2\tau}(R)
\le
\mathrm{EmitRate}
\le
p_\theta(R)+\beta\bigl(1-p_\theta(R)\bigr)
\le
p_\theta(R)+\beta.
\]
\end{theorem}

The theorem separates discovery from verification.
The verifier controls unsupported emission, while coverage depends on proposal mass in regions that can be certified under the available privacy budget.
In particular, candidates with \(f_s(R)\ge\theta+2\tau\) are emitted with probability at least \(1-\beta\).

\textbf{Tag-template proposer.}
Our primary proposer is designed around three principles: \emph{data independence}, so discovery incurs no privacy cost; \emph{matcher alignment}, so proposed statements can be evaluated reliably under the fixed support semantics; and \emph{threshold coverage}, so proposal mass is placed on statements likely to have sufficient support.
We instantiate these principles using a public schema over domain-relevant dimensions such as topic, aspect, and sentiment.
A sampled tuple is rendered through a fixed template, for example,
\[
\text{``Reviews about \texttt{\{topic\}} commonly mention
\texttt{\{aspect\}} in \texttt{\{sentiment\}} contexts.''}
\]
The schema, values, and templates are fixed before observing \(R\).
Controlled surface forms make matcher behavior easier to validate on public or otherwise non-sensitive examples, while schema granularity trades specificity against the probability of clearing the reporting threshold.
A public-LLM proposer provides another data-independent discovery mechanism; we compare the two in Appendix~\ref{app:proposer-comparison}.

\subsection{The Cost of Fine-Grained Certification}
\label{subsec:budget-coverage}

The single-statement analysis assumes that a privacy budget has already been assigned to one statement.
With \(m\) statements sharing a total budget \(\varepsilon_{\mathrm{tot}}\), one-sided honesty alone is insufficient to express useful coverage, since always abstaining is trivially honest.
We therefore call a mechanism \((w,\beta)\)-fine-grained if, for every statement \(s_j\),
\[
f_{s_j}(R)<\theta
\Longrightarrow
\Pr[j\in O]\le\beta,
\qquad
f_{s_j}(R)\ge\theta+w
\Longrightarrow
\Pr[j\in O]\ge1-\beta.
\]

\begin{theorem}[Cost of fine-grained certification]
\label{thm:fine-grained-cost}
Consider a worst-case regime in which the record universe can realize every binary support pattern over \(m\) statements, and assume \(\lceil n\theta\rceil+\lceil nw\rceil\le n\).
For fixed \(0<\beta<1/8\) and \(\varepsilon_{\mathrm{tot}}=o(m)\), any pure \(\varepsilon_{\mathrm{tot}}\)-DP mechanism satisfying \((w,\beta)\)-fine-grained certification must have
\[
w
=
\Omega\!\left(
\frac{m}{n\varepsilon_{\mathrm{tot}}}
\right).
\]
\end{theorem}
The finite-sample bound and packing proof are given in Appendix~B.
The richness assumption is a worst-case condition rather than a claim about every semantic support model.

Under equal allocation, sharply calibrated Laplace \textsc{PrivCert}-PF satisfies fine-grained certification with
\(
w_{\mathrm{PF}}
=
\frac{2m}{n\varepsilon_{\mathrm{tot}}}
\log\frac{1}{2\beta},
\)
matching the lower bound's dependence on \(m/(n\varepsilon_{\mathrm{tot}})\) in this regime for fixed \(\beta\).

Thus, in this regime, the linear scaling is not merely an artifact of
Laplace noise or equal budget allocation: in the worst-case overlapping-support regime, it is the cost of resolving many statement supports individually.
Other \textsc{PrivCert} instantiations can improve this trade-off when additional structure, joint sensitivity, or selective release can be exploited.

\section{Empirical Evaluation}
\label{sec:experiments}

We evaluate whether the \textsc{PrivCert} framework developed above is reflected in practice. 
Our experiments ask whether
(i) free-text DP baselines emit claims with low support,
(ii) \textsc{PrivCert}-PF controls unsupported emission,
(iii) its EmitRate follows Theorem~\ref{thm:precision-coverage}, and
(iv) alternative \textsc{PrivCert} instantiations realize the same reporting contract under different trade-offs.

\textbf{Metrics.}
For \textsc{PrivCert}-PF, we report UnsupportedEmit, the empirical
per-candidate failure rate controlled by Theorem~\ref{thm:honesty}, and EmitRate, the coverage quantity characterized by Theorem~\ref{thm:precision-coverage}.
We additionally report FalseEmission and MeanSupport as descriptive diagnostics; full definitions are given in Appendix~\ref{app:experimental-diagnostics}.

Additional analyses of proposer design, semantic-matcher robustness, and family-wise guarantees are reported in Appendices~\ref{app:proposer-comparison},
\ref{app:matcher-robustness}, and \ref{app:family-wise}, respectively.

\subsection{Synthetic Validation}
\label{subsec:synthetic-validation}

We first isolate the \textsc{PrivCert}-PF evidence mechanism from proposer and classifier effects using a finite statement universe with oracle support values.
Candidates are sampled uniformly, and we sweep
\(\varepsilon_{\mathrm{tot}}\in\{1,4,10,20\}\),
\(\theta\in\{0.05,0.1,0.2\}\),
\(m\in\{50,200,1000\}\), and
\(\beta\in\{0.01,0.05,0.1\}\) over five seeds.
The results are consistent with the one-sided calibration and Theorem~\ref{thm:precision-coverage} across the full grid.
UnsupportedEmit remains low under the prescribed one-sided calibration (maximum \(0.038\), mean \(0.002\)), and empirical EmitRate lies within the
interval predicted by Theorem~\ref{thm:precision-coverage} in all tested cells up to sampling slack.
The sweep also illustrates the effect of \(m\) under a fixed total privacy budget: increasing \(m\) evaluates more candidates in absolute number but reduces the privacy budget available to each support query, increasing the certification margin.
Because EmitRate is a fraction, increasing \(m\) does not itself increase the proposal-induced support mass \(p_a(R)\).
Full synthetic settings and diagnostics are reported in Appendix~\ref{app:metrics-synthetic}.

\textbf{Takeaway.}
Across the synthetic sweep, explicit certification reliably suppresses unsupported statements, while its coverage changes as predicted with the available privacy budget.

\subsection{Realistic Corpus Experiments}
\label{subsec:realistic-results}

We next evaluate whether \textsc{PrivCert}-PF controls unsupported emission in realistic text-reporting settings, and whether the evidence gap also appears in free-text DP outputs.
We use TAB (legal text)~\citep{pilan2022tab}, WildChat (conversations)~\citep{zhao2024wildchat}, and Yelp (reviews)~\citep{zhang2015character}.

\textbf{Support and proposer.}
We use a frozen public NLI cross-encoder to define the record-level support predicate and the tag-template proposer from Section~\ref{sec:pf-discovery}.
The matcher is fixed independently of \(R\), giving sensitivity at most \(1/|R|\).
Full support-model and calibration details are given in Appendix~\ref{app:support-model}.

\textbf{Baselines.}
We compare against no-data unfiltered generation, an aggregate-then-render document-level DP baseline, and token-level DP decoding with
\textsc{InvisibleInk}~\citep{vinod2025invisibleink}, all at the common target \((\varepsilon_{\mathrm{tot}},\delta)=(10,10^{-5})\).
Details are given in Appendix~\ref{app:support-model}.

\begin{table}[t]
\centering
\footnotesize
\caption{\textbf{FalseEmission on realistic datasets under the declared support semantics.}
For free-text methods this is a descriptive diagnostic, not a theorem-backed error rate. Lower is better.}
\label{tab:false-emission-main}
\begin{tabular}{lrrrrr}
\toprule
Dataset
& Unfiltered
& Doc-DP
& \shortstack[c]{\textsc{InvisibleInk}\\(TinyLlama)}
& \shortstack[c]{\textsc{InvisibleInk}\\(Qwen)}
& \textsc{PrivCert}-PF\\
\midrule
TAB      & 1.000 & 0.894 & 0.958 & 0.994 & \textbf{0.010} \\
WildChat & 0.947 & 0.729 & 0.647 & 0.619 & \textbf{0.014} \\
Yelp     & 0.943 & 0.803 & 0.877 & 0.856 & \textbf{0.000} \\
\bottomrule
\end{tabular}
\end{table}

\begin{table}[t]
\centering
\small
\caption{\textbf{Privacy-budget sweep on WildChat.}
Larger budgets improve EmitRate while preserving one-sided honesty.
Intervals use the tightened upper bound
\(p_\theta+\beta(1-p_\theta)\).}
\label{tab:privacy-budget-sweep}
\begin{tabular}{rrrrrlc}
\toprule
\(\varepsilon_{\mathrm{tot}}\) &
EmitRate \(\uparrow\) &
UnsupEmit \(\downarrow\) &
FalseEmit \(\downarrow\) &
MeanSupport \(\uparrow\) &
Interval & In? \\
\midrule
1  & 0.051 & 0.0217 & 0.295 & 0.275 & [0.000, 0.335] & \(\checkmark\) \\
4  & 0.180 & 0.0085 & 0.033 & 0.327 & [0.135, 0.335] & \(\checkmark\) \\
10 & 0.220 & 0.0046 & 0.014 & 0.288 & [0.162, 0.335] & \(\checkmark\) \\
20 & 0.274 & 0.0014 & 0.004 & 0.250 & [0.189, 0.335] & \(\checkmark\) \\
\bottomrule
\end{tabular}
\end{table}

Table~\ref{tab:false-emission-main} summarizes the main results.

\textbf{Free-text DP outputs often contain low-support claims under the declared support semantics.}
FalseEmission remains high for both document-level and token-level DP baselines.
For \textsc{InvisibleInk}, it ranges from 0.619 to 0.994 across datasets and model backbones.
Thus, privacy-preserving generation alone does not make statement occurrence an explicit support certificate.
Because the methods differ in output structure and candidate-generation procedure, Table~\ref{tab:false-emission-main} should be read as a diagnostic of the evidence gap under a common support measure, rather than as a causal estimate of the effect of certification or a general ranking of semantic quality.

\textbf{\textsc{PrivCert}-PF realizes the \textsc{PrivCert} contract.}
Across the three datasets, UnsupportedEmit lies in \([0.000,0.008]\), below the target \(\beta=0.05\).
This is the per-candidate quantity controlled by Theorem~\ref{thm:honesty}.
The same behavior holds across proposer ablations (Appendix~\ref{app:proposer-comparison}).
We further evaluate \textsc{PrivCert}-PF with a frozen Qwen3-8B judge as an alternative support model in Appendix~\ref{app:matcher-robustness}.

\textbf{EmitRate is consistent with the predicted trade-off.}
The observed aggregate EmitRates and all 15 per-seed realizations fall within the interval predicted for the expectation by Theorem~\ref{thm:precision-coverage}.
For the aggregate results, TAB: \(0.404\in[0.356,0.532]\), WildChat: \(0.220\in[0.162,0.335]\), and Yelp: \(0.749\in[0.620,0.809]\).
These finite-run checks are empirical consistency tests rather than additional per-run guarantees.

\textbf{Privacy-budget sweep.}
We evaluate the coverage--noise trade-off on WildChat while keeping the tag-template proposer, support model, threshold, and per-candidate honesty target fixed.
We vary \(\varepsilon_{\mathrm{tot}}\in\{1,4,10,20\}\) with \(\theta=0.05\) and \(\beta=0.05\).
As \(\varepsilon_{\mathrm{tot}}\) increases, support estimates become less noisy and EmitRate rises from \(0.051\) to \(0.274\).
UnsupportedEmit remains below \(\beta=0.05\) throughout.
The results are consistent with the inverse-budget dependence of the certification margin characterized in Section~\ref{subsec:budget-coverage}.

\textbf{Whole-report control.}
The guarantee above applies to each candidate separately: an unsupported candidate is emitted with probability at most \(\beta\).
When many candidates are evaluated, these risks can accumulate, so this does not bound the probability that the final report contains any unsupported statement.
To obtain such a report-level guarantee, we allocate a total failure budget
\(\beta_{\mathrm{tot}}\) across candidates so that
\(
\Pr[\exists s\in O:\ f_s(R)<\theta]
\le
\beta_{\mathrm{tot}}.
\)
With \(\beta_{\mathrm{tot}}=0.05\), EmitRate decreases from \(0.404\) to \(0.357\) on TAB, \(0.220\) to \(0.169\) on WildChat, and \(0.749\) to \(0.565\) on Yelp (Appendix~\ref{app:family-wise}).

\textbf{Takeaway.}
On realistic corpora, \textsc{PrivCert}-PF maintains low unsupported emission and the predicted coverage behavior, while free-text DP outputs frequently contain low-support claims under the same declared evidence semantics.

\subsection{Alternative \textsc{PrivCert} Instantiations}
\label{subsec:alternative-evidence}

\begin{table}[t]
\centering
\footnotesize
\caption{\textbf{Matched comparisons of \textsc{PrivCert} instantiations.}
Entries show EmitRate. All UnsupportedEmit values remain below \(\beta=0.05\).
Comparisons are matched within each candidate space.}
\label{tab:alternative-mechanisms}
\begin{tabular}{lccccc}
\toprule
& \multicolumn{2}{c}{Sampled \(m=200\)}
& \multicolumn{3}{c}{Exact-unique} \\
\cmidrule(lr){2-3}\cmidrule(lr){4-6}
Dataset
& \textsc{PrivCert}-Hist
& \textsc{PrivCert}-PF
& \textsc{PrivCert}-SVT
& \textsc{PrivCert}-PF
& \textsc{PrivCert}-PF (Gaus) \\
\midrule
TAB
& 0.199 & 0.404
& 0.400 & 0.422 & 0.444 \\
WildChat
& 0.183 & 0.220
& 0.151 & 0.151 & 0.162 \\
Yelp
& 0.390 & 0.749
& 0.796 & 0.816 & 0.833 \\
\bottomrule
\end{tabular}
\end{table}

We next evaluate alternative instantiations of the \textsc{PrivCert} framework.
\textsc{PrivCert}-Hist uses a contribution-bounded joint support release, \textsc{PrivCert}-SVT combines private sparse selection with subsequent certification, and a Gaussian variant of \textsc{PrivCert}-PF replaces the
pure-DP Laplace mechanism with an analytic Gaussian mechanism under \((\varepsilon,\delta)\)-DP.
Table~\ref{tab:alternative-mechanisms} reports EmitRate; all corresponding UnsupportedEmit values remain below the target \(\beta=0.05\).
Full mechanism definitions and privacy accounting are given in Appendix~\ref{app:alternative-backends}.

\textsc{PrivCert}-Hist changes sensitivity through contribution bounding, while \textsc{PrivCert}-SVT privately selects a bounded candidate set before certification, so Theorem~\ref{thm:precision-coverage} does not apply unchanged. 
The Gaussian variant changes only the DP backend.
Together, these matched comparisons show that the \textsc{PrivCert} contract is compatible with different evidence constructions and privacy mechanisms.

\textbf{A boundary case: private proxies.}
An Aug-PE~\citep{xie2024differentially}-style DP synthetic-data backend illustrates a boundary of the \textsc{PrivCert} contract: \(f_s(\widetilde R)\ge\theta\) does not imply \(f_s(R)\ge\theta\).
At \(\theta=0.10\), the synthetic backend attains similar EmitRate (\(0.676\) versus \(0.695\) for direct certification) but an original-data UnsupportedEmit of \(0.092\), compared with \(0.001\) for \textsc{PrivCert}-PF.
Original-data certification therefore requires an additional fidelity guarantee (Appendix~\ref{app:synthetic-data}).

\textbf{Takeaway.}
The \textsc{PrivCert} contract is not tied to a particular DP mechanism, but what is certified still depends on the reference data: support in a private proxy does not by itself certify support in the original dataset.
\section{Discussion and Limitations}

\textbf{What \textsc{PrivCert} certifies.}
\textsc{PrivCert} guarantees are relative to the declared reference data and support function \(f_s\), not matcher-independent semantic truth.
Different support models can preserve mechanism-level honesty while changing coverage and certified content (Appendix~\ref{app:matcher-robustness}).
Support measured on DP synthetic data certifies only that proxy absent an additional fidelity guarantee, and statements never proposed cannot be certified.

\textbf{Honesty does not imply completeness.}
Privacy noise, conservative certification, and limited discovery may cause abstention, and statement-level certificates do not guarantee logical consistency of the complete report.
\section{Conclusion}

We formulated the \emph{evidence gap} in privacy-preserving text generation: differential privacy protects individual records but does not specify what evidence a released statement carries about the underlying data.
We introduced \textsc{PrivCert}, a reporting framework that makes statement support explicit through privacy-preserving certificates and emit-or-abstain decisions.
We provided theoretical grounding for this framework by characterizing fundamental privacy--honesty limits, showing that sharply calibrated \textsc{PrivCert}-PF operates near the single-statement optimum, and establishing a worst-case cost for fine-grained multi-statement certification.
Experiments across multiple datasets and DP mechanisms empirically characterize the evidence gap and demonstrate alternative realizations of the reporting contract while clarifying its scope and boundaries.
Together, these results recast privacy-preserving reporting as an \emph{evidence-design problem}: not only \emph{how to generate private text}, but \emph{what a private report can certify about the data, and at what privacy cost}.

\bibliography{references}
\bibliographystyle{abbrv}

\clearpage

\appendix
\section*{Appendix}
\section{Proofs for Section~\ref{sec:limits}}
\label{app:proofs-limits}

\subsection{Proof of Theorem~\ref{thm:testing-lb} and Support-Estimation Consequence}

Let
\(
P_i:=\mathcal L(\mathcal A(q,R_i)),
\,
t:=e^{k\varepsilon}.
\)
By group privacy, for every measurable event \(E\),
\[
P_0(E)\le tP_1(E),
\qquad
P_1(E)\le tP_0(E).
\]

To bound total variation, fix an event \(B\) with \(P_0(B)\ge P_1(B)\), and write \(p:=P_0(B)\), \(q:=P_1(B)\), and \(\Delta:=p-q\).
Applying the two group-privacy inequalities to \(B\) and \(B^c\) gives
\[
\Delta\le (t-1)q,
\qquad
\Delta\le \frac{t-1}{t}(1-q).
\]
The minimum of these two bounds is maximized at \(q=1/(t+1)\), hence
\[
d_{\mathrm{TV}}(P_0,P_1)
\le
\frac{t-1}{t+1}
=
\frac{e^{k\varepsilon}-1}{e^{k\varepsilon}+1}.
\]
The standard binary-testing identity then gives, for every possibly
randomized test \(\varphi:\mathcal Y\to[0,1]\),
\[
\mathbb E_{P_0}[\varphi(Y)]
+
\mathbb E_{P_1}[1-\varphi(Y)]
\ge
1-d_{\mathrm{TV}}(P_0,P_1)
\ge
\frac{2}{e^{k\varepsilon}+1},
\]
which proves Theorem~\ref{thm:testing-lb}.

For the support-estimation consequence, let
\(a_i:=f_s(R_i)\) and
\(\Delta_s:=|a_1-a_0|\), and assume without loss of generality that
\(a_1>a_0\).
Given any estimator \(\hat f_s:\mathcal Y\to[0,1]\), define
\[
\varphi(Y)
:=
\mathbf 1
\left\{
\hat f_s(Y)\ge\frac{a_0+a_1}{2}
\right\}.
\]
If \(\varphi(Y)=1\) under \(P_0\), or \(\varphi(Y)=0\) under \(P_1\), the estimation error is at least \(\Delta_s/2\).
Therefore,
\[
\begin{aligned}
&
\mathbb E_{P_0}
\bigl[|\hat f_s(Y)-a_0|\bigr]
+
\mathbb E_{P_1}
\bigl[|\hat f_s(Y)-a_1|\bigr]
\\
&\qquad\ge
\frac{\Delta_s}{2}
\left(
\mathbb E_{P_0}[\varphi(Y)]
+
\mathbb E_{P_1}[1-\varphi(Y)]
\right).
\end{aligned}
\]
Since the maximum of two nonnegative quantities is at least half their sum,
Theorem~\ref{thm:testing-lb} implies
\[
\max_{i\in\{0,1\}}
\mathbb E_{Y\sim P_i}
\left[
|\hat f_s(Y)-a_i|
\right]
\ge
\frac{\Delta_s}{2(e^{k\varepsilon}+1)}.
\]

\subsection{Proof of Theorem~\ref{thm:sharp-frontier}}
\label{app:sharp-frontier}

Fix a statement \(s\), and write
\[
p_{\mathcal A}(R)
:=
\Pr[a_s=\mathrm{emit}\mid R].
\]
For adjacent datasets \(R\sim R'\), differential privacy applied to the emission event and to the complementary abstention event gives
\[
p_{\mathcal A}(R')
\le
e^\varepsilon p_{\mathcal A}(R),
\]
and
\[
p_{\mathcal A}(R')
\le
1-e^{-\varepsilon}
\bigl(1-p_{\mathcal A}(R)\bigr).
\]
Thus, if one additional supporting record is introduced,
\[
p_{\mathcal A}(R')
\le
g_\varepsilon\!\left(p_{\mathcal A}(R)\right),
\]
where
\[
g_\varepsilon(p)
:=
\min\left\{
e^\varepsilon p,\,
1-e^{-\varepsilon}(1-p)
\right\}.
\]
The map \(g_\varepsilon\) is nondecreasing.

\textbf{Upper bound.}
Let
\(
\pi_0:=\beta,
\,
\pi_{k+1}:=g_\varepsilon(\pi_k).
\)
We prove by induction on \(d_s(R)\) that
\[
p_{\mathcal A}(R)\le\pi_{d_s(R)}.
\]

For \(d_s(R)=0\), the statement is unsupported, so the honesty condition gives
\[
p_{\mathcal A}(R)\le\beta=\pi_0.
\]
For \(d_s(R)=k+1\), replace one supporting record by a fixed non-supporting record to obtain an adjacent dataset \(R^{-}\) with
\(d_s(R^{-})=k\).
Then
\[
p_{\mathcal A}(R)
\le
g_\varepsilon\!\left(p_{\mathcal A}(R^{-})\right)
\le
g_\varepsilon(\pi_k)
=
\pi_{k+1},
\]
where the second inequality uses the induction hypothesis and monotonicity of \(g_\varepsilon\).

\textbf{Tightness.}
Consider the count-based mechanism \(\mathcal A^\star\) that emits \(s\) with probability
\[
p^\star(R):=\pi_{d_s(R)}
\]
and otherwise abstains.
Adjacent datasets have \(d_s\)-values differing by at most one.
For consecutive \(k\) and \(k+1\), the recurrence gives
\[
\pi_{k+1}\le e^\varepsilon\pi_k,
\qquad
1-\pi_k
\le
e^\varepsilon(1-\pi_{k+1}).
\]
Because \(\pi_{k+1}\ge\pi_k\), the reverse-direction DP inequalities hold as well.
Hence the binary mechanism is \(\varepsilon\)-DP.

For unsupported datasets \(d_s(R)=0\), so \(p^\star(R)=\beta\), establishing \(\beta\)-honesty.
Thus the frontier is attained simultaneously at every support count.
Attaching \(c_s=\theta\) upon emission and \(c_s=0\) otherwise also satisfies the \textsc{PrivCert} certificate-validity condition.

\textbf{Half-emission consequence.}
Since
\(
g_\varepsilon(p)\le e^\varepsilon p,
\)
we have
\(
\pi_k\le\beta e^{k\varepsilon}.
\)
Therefore, emission probability at least \(1/2\) requires
\[
d_s(R)
\ge
\left\lceil
\frac{1}{\varepsilon}
\log\frac{1}{2\beta}
\right\rceil,
\]
which is the margin lower bound used in Corollary~\ref{cor:pf-near-optimal}.

\subsection{Extension to Approximate Differential Privacy}
\label{app:approx-dp}

The local indistinguishability results in Section~\ref{subsec:implicit-limits} extend to \((\varepsilon,\delta)\)-DP.
For datasets \(R_0,R_1\) with \(d(R_0,R_1)\le k\), define
\[
\delta_k
:=
\delta
\sum_{\ell=0}^{k-1}e^{\ell\varepsilon},
\qquad
t:=e^{k\varepsilon}.
\]
Approximate group privacy gives
\[
P_0(E)\le tP_1(E)+\delta_k,
\qquad
P_1(E)\le tP_0(E)+\delta_k
\]
for every measurable event \(E\).

For an event \(B\) with \(P_0(B)\ge P_1(B)\), writing \(q:=P_1(B)\) and \(\Delta:=P_0(B)-P_1(B)\), these inequalities imply
\[
\Delta
\le
(t-1)q+\delta_k,
\qquad
\Delta
\le
\frac{t-1}{t}(1-q)+\frac{\delta_k}{t}.
\]
Maximizing the smaller of the two bounds yields
\[
d_{\mathrm{TV}}(P_0,P_1)
\le
\min\left\{
1,\,
\frac{t-1+2\delta_k}{t+1}
\right\}.
\]
Consequently, every possibly randomized binary test satisfies
\[
\mathbb E_{P_0}[\varphi(Y)]
+
\mathbb E_{P_1}[1-\varphi(Y)]
\ge
\left[
\frac{2(1-\delta_k)}
{e^{k\varepsilon}+1}
\right]_+,
\]
where \([x]_+:=\max\{x,0\}\).

Applying the same two-point reduction as above, with
\(a_i:=f_s(R_i)\) and
\(\Delta_s:=|a_1-a_0|\), gives
\[
\max_{i\in\{0,1\}}
\mathbb E_{Y\sim P_i}
\left[
|\hat f_s(Y)-a_i|
\right]
\ge
\left[
\frac{\Delta_s(1-\delta_k)}
{2(e^{k\varepsilon}+1)}
\right]_+.
\]
Setting \(\delta=0\) recovers the pure-DP bounds.
As in the pure-DP case, these local bounds may become weak when \(k\varepsilon\) or the accumulated slack \(\delta_k\) is large.
\section{Additional Details and Proofs for Section~\ref{sec:method}}
\label{app:proofs-method}

\subsection{\textsc{PrivCert}-PF Algorithm}
\label{app:pf-algorithm}

\begin{algorithm}[H]
\caption{\textsc{PrivCert}-PF}
\label{alg:proposal-filter}
\begin{algorithmic}[1]
\REQUIRE Query \(q\), dataset \(R\), proposal distribution \(\Pi_q\),
number of candidates \(m\), threshold \(\theta\), margin \(\tau\),
per-query privacy budget \(\varepsilon_0\)
\STATE Sample \(s_1,\dots,s_m \overset{\mathrm{i.i.d.}}{\sim}\Pi_q\)
\FOR{\(j=1,\dots,m\)}
    \STATE Compute
    \(\widehat f_{s_j}(R)=f_{s_j}(R)+Z_j\)
    using an \(\varepsilon_0\)-DP support query
    \STATE Set \(L_{s_j}=\widehat f_{s_j}(R)-\tau\)
    \IF{\(L_{s_j}\ge\theta\)}
        \STATE Emit \(s_j\) with certificate \(L_{s_j}\)
    \ELSE
        \STATE Abstain
    \ENDIF
\ENDFOR
\STATE Optionally render emitted statements as post-processing
\end{algorithmic}
\end{algorithm}

\subsection{Privacy and Honesty}

\textbf{Privacy accounting.}
The proposal distribution \(\Pi_q\) is independent of \(R\), so candidate sampling incurs no privacy cost.
Conditioned on the sampled candidates, if support query \(j\) is
\(\varepsilon_j\)-DP, sequential composition gives
\(
\left(\sum_{j=1}^m \varepsilon_j\right)\text{-DP}.
\)
The certificates, emit-or-abstain decisions, and any final rendering are post-processing and incur no additional privacy loss.

\textbf{Proof of Theorem~\ref{thm:honesty}.}
For an evaluated statement \(s\),
\[
L_s=f_s(R)+Z-\tau.
\]
Hence
\[
L_s>f_s(R)
\quad\Longrightarrow\quad
Z>\tau,
\]
so
\[
\Pr[L_s>f_s(R)]
\le
\Pr[Z\ge\tau]
\le
\beta.
\]
If \(f_s(R)<\theta\) and \(s\) is emitted, then
\[
f_s(R)+Z-\tau\ge\theta
\quad\Longrightarrow\quad
Z\ge\tau+\theta-f_s(R)>\tau,
\]
and therefore
\[
\Pr[\mathrm{Emit}_s]
\le
\Pr[Z\ge\tau]
\le
\beta.
\]
This proves certificate validity and per-candidate one-sided honesty.

\textbf{Report-level control.}
Let \(B_j\) denote the event that candidate \(j\) is unsupported and emitted.
Conditioned on the sampled candidates, Theorem~\ref{thm:honesty} gives \(\Pr[B_j]\le\beta_j\).
Thus
\[
\Pr\!\left[\bigcup_{j=1}^m B_j\right]
\le
\sum_{j=1}^m \beta_j.
\]
In particular, uniform \(\beta_j=\beta\) gives a report-level bound \(m\beta\).
No independence between candidate-level emission events is required.

\subsection{Proof of Theorem~\ref{thm:precision-coverage}}

By linearity of expectation and identical sampling from \(\Pi_q\),
\[
\mathrm{EmitRate}
=
\Pr_{s\sim\Pi_q,Z}
\bigl[
f_s(R)+Z-\tau\ge\theta
\bigr].
\]

For the upper bound, if \(f_s(R)<\theta\),
Theorem~\ref{thm:honesty} gives emission probability at most \(\beta\); otherwise we use the trivial bound \(1\).
Hence
\[
\begin{aligned}
\mathrm{EmitRate}
&\le
p_\theta(R)
+
\beta\bigl(1-p_\theta(R)\bigr)
\\
&\le
p_\theta(R)+\beta.
\end{aligned}
\]

For the lower bound, if
\[
f_s(R)\ge\theta+2\tau,
\]
then
\[
\theta+\tau-f_s(R)\le-\tau,
\]
so
\[
\Pr[s\text{ is emitted}\mid s]
=
\Pr[Z\ge\theta+\tau-f_s(R)]
\ge
\Pr[Z\ge-\tau]
\ge
1-\beta.
\]
Averaging over \(s\sim\Pi_q\) gives
\[
\mathrm{EmitRate}
\ge
(1-\beta)p_{\theta+2\tau}(R).
\]
Combining the bounds proves
\[
(1-\beta)p_{\theta+2\tau}(R)
\le
\mathrm{EmitRate}
\le
p_\theta(R)+\beta(1-p_\theta(R))
\le
p_\theta(R)+\beta.
\]

\subsection{Proof of Corollary~\ref{cor:pf-near-optimal}}

Let
\[
A
:=
\frac{1}{\varepsilon_s}
\log\frac{1}{2\beta},
\qquad
k^\star:=\lceil A\rceil.
\]
Theorem~\ref{thm:sharp-frontier} implies that emission probability \(1/2\) requires distance at least \(k^\star\).

For Laplace \textsc{PrivCert}-PF, define
\[
\rho:=\lceil n\theta\rceil-n\theta\in[0,1).
\]
At distance \(d\) from the largest unsupported support count,
\[
f_s(R)-\theta
=
\frac{d-1+\rho}{n}.
\]
With
\[
\tau
=
\frac{1}{n\varepsilon_s}
\log\frac{1}{2\beta}
=
\frac{A}{n},
\]
the emission probability is at least \(1/2\) iff \(f_s(R)-\theta\ge\tau\), equivalently
\[
d\ge A+1-\rho.
\]
Thus the first such distance is
\[
d_{\mathrm{PF}}
=
\lceil A+1-\rho\rceil
\in
\{k^\star,k^\star+1\}.
\]
The corresponding support count is \(\lceil n\theta+A\rceil\), which is at most \(n\) by assumption; hence this distance is attainable.

\subsection{Proof of Theorem~\ref{thm:fine-grained-cost}}
\label{app:proof-fine-grained}

We prove the finite-sample bound
\[
nw
\ge
\frac{
\kappa m\log 2+\log(1-8\beta)
}{
\varepsilon_{\mathrm{tot}}
}
-2,
\qquad
\kappa:=1-H_2(1/4),
\]
where \(H_2\) is binary entropy in bits.
For fixed \(0<\beta<1/8\) and any regime with \(\varepsilon_{\mathrm{tot}}=o(m)\), this implies
\[
w
=
\Omega\!\left(
\frac{m}{n\varepsilon_{\mathrm{tot}}}
\right).
\]

\textbf{Hard-instance construction.}
Let
\[
t:=\lceil n\theta\rceil-1,
\qquad
h:=\lceil n(\theta+w)\rceil,
\qquad
d:=h-t.
\]
Then
\[
\frac{t}{n}<\theta,
\qquad
\frac{h}{n}\ge\theta+w,
\]
and the feasibility assumption ensures \(h\le n\).

By the richness assumption, for every \(v\in\{0,1\}^m\) there is a record \(r_v\) whose support pattern is \(v\).
Let \(r_{\mathbf 0}\) and \(r_{\mathbf 1}\) denote the all-zero and all-one patterns.
Construct \(R^{(0)}\) from \(t\) copies of \(r_{\mathbf 1}\) and \(n-t\) copies of \(r_{\mathbf 0}\).
For each \(v\), obtain \(R^{(v)}\) by replacing \(d\) copies of \(r_{\mathbf 0}\) by \(r_v\).
Then
\[
d(R^{(0)},R^{(v)})\le d,
\]
and, for every \(j\),
\[
v_j=0
\Longrightarrow
f_{s_j}(R^{(v)})=\frac{t}{n}<\theta,
\qquad
v_j=1
\Longrightarrow
f_{s_j}(R^{(v)})=\frac{h}{n}\ge\theta+w.
\]

\textbf{Packing and decoding.}
By the Gilbert packing bound, there exists \(\mathcal C\subseteq\{0,1\}^m\) with pairwise Hamming distance greater than \(m/4\) and
\[
|\mathcal C|
\ge
\frac{2^m}{
\sum_{i=0}^{\lfloor m/4\rfloor}\binom mi
}
\ge
2^{\kappa m}.
\]

Let
\[
X_j:=\mathbf 1\{j\in O\},
\qquad
X=(X_1,\dots,X_m).
\]
For \(v\in\mathcal C\), fine-grained certification gives
\[
\Pr[X_j\neq v_j\mid R^{(v)}]\le\beta
\]
for every coordinate \(j\), and therefore
\[
\mathbb E[d_H(X,v)\mid R^{(v)}]
\le
m\beta.
\]
By Markov's inequality,
\[
\Pr\!\left[
d_H(X,v)<\frac{m}{8}
\,\middle|\,
R^{(v)}
\right]
\ge
1-8\beta.
\]
Define
\[
E_v
:=
\left\{
x\in\{0,1\}^m:
d_H(x,v)<\frac{m}{8}
\right\}.
\]
Because codewords in \(\mathcal C\) are more than \(m/4\) apart, the events \(\{E_v:v\in\mathcal C\}\) are pairwise disjoint.

\textbf{Applying differential privacy.}
Group privacy and \(d(R^{(0)},R^{(v)})\le d\) imply
\[
\Pr_{R^{(0)}}[X\in E_v]
\ge
e^{-d\varepsilon_{\mathrm{tot}}}
\Pr_{R^{(v)}}[X\in E_v]
\ge
(1-8\beta)e^{-d\varepsilon_{\mathrm{tot}}}.
\]
Summing over the disjoint decoding events,
\[
1
\ge
\sum_{v\in\mathcal C}
\Pr_{R^{(0)}}[X\in E_v]
\ge
|\mathcal C|
(1-8\beta)
e^{-d\varepsilon_{\mathrm{tot}}},
\]
hence
\[
d\varepsilon_{\mathrm{tot}}
\ge
\log|\mathcal C|
+
\log(1-8\beta)
\ge
\kappa m\log 2+\log(1-8\beta).
\]

Finally,
\[
d
=
\lceil n(\theta+w)\rceil
-\lceil n\theta\rceil+1
\le
\lceil nw\rceil+1
\le
nw+2.
\]
Therefore
\[
nw
\ge
\frac{
\kappa m\log 2+\log(1-8\beta)
}{
\varepsilon_{\mathrm{tot}}
}
-2,
\]
which proves Theorem~\ref{thm:fine-grained-cost}.
\section{Additional Experimental Details and Diagnostics}
\label{app:experimental-diagnostics}

This appendix provides implementation details and additional diagnostics for
Section~\ref{sec:experiments}. We first describe the evaluation protocol, support model, proposal mechanisms, and privacy accounting, and then report additional comparisons that clarify the scope of the \textsc{PrivCert} guarantees.

\subsection{Experimental Protocol and Metrics}
\label{app:metrics-synthetic}

\textbf{Metrics.}
For each proposal-and-filter repetition \(k\), we compute
\[
\widehat{\mathrm{UnsupportedEmit}}_k
=
\frac{
\sum_j
\mathbf 1\{
s_{kj}\text{ emitted},
f_{s_{kj}}(R)<\theta
\}
}{
\sum_j
\mathbf 1\{
f_{s_{kj}}(R)<\theta
\}
},
\]
and report the mean across the \(S\) independent repetitions.
This is the empirical analogue of the per-candidate quantity controlled by Theorem~\ref{thm:honesty}.

We also report the emitted fraction
\[
\widehat{\mathrm{EmitRate}}_k
=
\frac{1}{m}
\sum_j
\mathbf 1\{s_{kj}\text{ emitted}\},
\]
averaged across repetitions, together with
\[
\widehat{\mathrm{FalseEmission}}
=
\frac{
\sum_j
\mathbf 1\{
s_j\text{ emitted},
f_{s_j}(R)<\theta
\}
}{
\sum_j
\mathbf 1\{s_j\text{ emitted}\}
},
\qquad
\widehat{\mathrm{MeanSupport}}
=
\frac{
\sum_j
\mathbf 1\{s_j\text{ emitted}\}
f_{s_j}(R)
}{
\sum_j
\mathbf 1\{s_j\text{ emitted}\}
}.
\]
FalseEmission and MeanSupport are descriptive diagnostics rather than quantities directly controlled by Theorem~\ref{thm:honesty}.
For free-text baselines, they are computed over parsed released claims.

\textbf{Synthetic setup.}
The synthetic experiment isolates the private filter from semantic-matcher and proposal effects. We construct a finite statement universe \(\mathcal S=\{s_1,\dots,s_M\}\), assign each statement an oracle support value, and sample candidates uniformly from \(\mathcal S\).
Support values are drawn from a two-component mixture containing lower- and higher-support statements.

We sweep
\(
\varepsilon_{\mathrm{tot}}\in\{1,4,10,20\},
\,
\theta\in\{0.05,0.1,0.2\},
\,
m\in\{50,200,1000\},
\,
\beta\in\{0.01,0.05,0.1\},
\)
with five seeds per configuration.

\textbf{Repetitions and hardware.}
Realistic proposal-and-filter experiments use five seeds per dataset--proposer pair and \(m=200\) candidates per seed.
Most free-text baselines use three seeds; the WildChat/\textsc{InvisibleInk} Qwen configuration uses five.
Experiments were run on Microsoft Azure VMs with NVIDIA H100 GPUs. Proposal-and-filter logs record the candidate, true operational support, noise draw, certificate, and emission decision, which are sufficient to reconstruct the reported metrics and theorem checks.

\subsection{Support Evaluation, Proposals, and Privacy Accounting}
\label{app:support-model}

\textbf{Free-text claim extraction.}
Free-text baselines do not provide a fixed candidate pool or explicit emit-or-abstain decisions. We therefore extract atomic asserted statements from each released report using a fixed claim parser and score the resulting claims with the same support function used for \textsc{PrivCert}-PF.
Claim extraction operates only on already released text and is therefore post-processing of the DP mechanism.

\textbf{Operational support model.}
Our default support predicate is a frozen public NLI cross-encoder, \texttt{cross-encoder/nli-deberta-v3-base}~\citep{he2021deberta}.
For a record \(r\) and statement \(s\), \(r\) is used as the premise and \(s\) as the hypothesis.
Inputs are tokenized with the native tokenizer and truncated to a maximum joint length of \(384\) wordpiece tokens.

Let \(\ell_c(r,s)\) be the NLI logit for class \(c\). We define
\[
\sigma(r,s)
=
\mathbf 1
\left[
\arg\max_c\ell_c(r,s)=\text{entailment}
\right],
\qquad
f_s(R)
=
\frac{1}{|R|}
\sum_{r\in R}
\sigma(r,s).
\]
The matcher is fixed independently of the private corpus.
Because replacing one record changes at most one binary summand,
\(
\Delta_{f_s}
\le
\frac{1}{|R|}
\)
under replace-one adjacency.
This sensitivity statement is independent of the semantic accuracy of the matcher; all guarantees are relative to the resulting operational support function.

\textbf{Proposal mechanisms.}
The tag-template proposer samples structured values from a public schema over dimensions such as topic, aspect, and sentiment, and renders them through fixed templates.
The public-LLM proposer instead receives only a public task description and generates candidate dataset-level claims.
Both proposal mechanisms are fixed independently of \(R\), so they incur no privacy cost.
The proposer affects discovery and coverage but not the per-candidate honesty guarantee of the subsequent private verifier.

\textbf{Privacy accounting.}
All privacy-constrained methods in the realistic experiments target
\(
(\varepsilon_{\mathrm{tot}},\delta)
=
(10,10^{-5}).
\)
Pure-DP mechanisms have \(\delta=0\) and therefore also satisfy this target.

For proposal-and-filter,
\[
\varepsilon_0
=
\frac{\varepsilon_{\mathrm{tot}}}{m},
\qquad
Z_j\sim\operatorname{Lap}(b),
\qquad
b
=
\frac{m\Delta_f}{\varepsilon_{\mathrm{tot}}}.
\]
The main experiments use the conservative margin
\[
\tau
=
b\log\frac{1}{\beta},
\]
which makes both required one-sided Laplace tail probabilities at most
\(\beta\).

The document-level baseline privately releases a fixed set of aggregate statistics and renders them with a public LLM. 
\textsc{InvisibleInk} uses its published private-decoding accountant; the zCDP budget is calibrated to the common \((10,10^{-5})\)-DP target, with parallel composition across disjoint reference batches.

\begin{table}[t]
\centering
\small
\caption{\textbf{Privacy accounting for the main-table methods.}
Pure-DP mechanisms also satisfy the displayed approximate-DP target.}
\label{tab:privacy-accounting}
\begin{tabular}{lllll}
\toprule
Method & DP type & Composition unit & Sensitivity & Final guarantee \\
\midrule
\textsc{PrivCert}-PF
& pure DP
& \(m\) support queries
& \(\Delta_f\le1/|R|\)
& \(\varepsilon_{\mathrm{tot}}\)-DP \\
Document-level DP
& pure DP
& \(d_{\mathrm{doc}}\) aggregates
& \(\Delta_{\mathrm{doc}}\)
& \(\varepsilon_{\mathrm{tot}}\)-DP \\
\textsc{InvisibleInk}
& zCDP
& private decoding
& DClip sensitivity
& \((\varepsilon_{\mathrm{tot}},10^{-5})\)-DP \\
Unfiltered LLM
& none
& no private-data access
& n/a
& reference only \\
\bottomrule
\end{tabular}
\end{table}

\subsection{Alternative \textsc{PrivCert} Mechanisms}
\label{app:alternative-backends}

Section~\ref{subsec:alternative-evidence} compares \textsc{PrivCert}-PF with alternative evidence constructions while holding the candidate space, support semantics, threshold, and honesty target fixed.

\textbf{Contribution-bounded histogram.}
For candidates \(s_1,\dots,s_m\), define
\[
\mathbf f(R)
=
\frac{1}{n}
\sum_{r\in R}
\bigl(
\sigma(r,s_1),\dots,\sigma(r,s_m)
\bigr).
\]
Its \(\ell_1\)-sensitivity is at most \(m/n\).
If each record is instead allowed to contribute to at most \(c\) positive coordinates, the capped vector satisfies
\[
\Delta_1(\widetilde{\mathbf f}_c)
\le
\frac{\min(m,2c)}{n},
\]
and therefore uses Laplace scale
\[
b_c
=
\frac{\min(m,2c)}
{n\varepsilon_{\mathrm{tot}}}.
\]
Contribution bounding can reduce noise but may remove genuine support.
Because only positive contributions are removed, \(\widetilde f_{c,s}(R)\le f_s(R)\), so the resulting lower certificates remain conservative with respect to the original support function.
The main comparison uses \(c=16\).

\textbf{Sparse-vector select-then-certify.}
The SVT variant allocates
\(
\varepsilon_{\mathrm{tot}}
=
\varepsilon_{\mathrm{select}}
+
\varepsilon_{\mathrm{cert}},
\)
uses the first budget to select at most \(k\) candidates, and certifies each selected candidate independently with budget \(\varepsilon_{\mathrm{cert}}/k\).
The final certification stage therefore retains the one-sided honesty argument of Theorem~\ref{thm:honesty}.
Theorem~\ref{thm:precision-coverage} does not apply unchanged because SVT introduces dependent selection, a shared noisy threshold, early stopping, and a cap on positive outputs. 
We sweep
\(
k\in\{5,10,20,40\},
\qquad
\frac{\varepsilon_{\mathrm{select}}}
{\varepsilon_{\mathrm{tot}}}
\in\{0.25,0.50,0.75\}.
\)

\textbf{Gaussian backend.}
To verify that the reporting interface is not specific to pure-DP Laplace noise, we also use
\[
\widehat f_s(R)
=
f_s(R)+Z_s,
\qquad
Z_s\sim\mathcal N(0,\sigma_G^2),
\]
with \(\sigma_G\) calibrated so that the complete multi-statement release satisfies \((\varepsilon_{\mathrm{tot}},\delta)=(10,10^{-5})\)-DP.
The certificate margin is
\[
\tau_G
=
\sigma_G\Phi^{-1}(1-\beta),
\]
so that
\[
\Pr[Z_s\ge\tau_G]=\beta.
\]
The same one-sided argument as in Theorem~\ref{thm:honesty} therefore applies.

These alternatives demonstrate that the \textsc{PrivCert} contract is not specific to a particular DP primitive; their coverage differs because they exploit different sensitivity and selection structures.

\subsection{Proposer Comparison}
\label{app:proposer-comparison}

We hold the support model, private verifier, privacy budget, and number of candidates fixed and change only the proposal distribution.
Both proposers use
\(
(\varepsilon_{\mathrm{tot}},\delta)=(10,10^{-5}),
\,
m=200,
\)
with five seeds per dataset.

\begin{table}[t]
\centering
\footnotesize
\caption{\textbf{Proposer comparison on realistic corpora.}
Only the proposal distribution changes; the private verifier is fixed.}
\label{tab:proposer-comparison}
\begin{tabular}{llrrrrr}
\toprule
Dataset & Proposer & \#out & EmitRate & UnsupEmit &
FalseEmit & MeanSupport \\
\midrule
\multirow{2}{*}{TAB}
 & public LLM   &   3 & 0.003 & 0.003 & 1.000 & 0.006 \\
 & tag-template & 404 & 0.404 & 0.008 & 0.010 & 0.545 \\
\midrule
\multirow{2}{*}{WildChat}
 & public LLM   &  24 & 0.024 & 0.001 & 0.042 & 0.132 \\
 & tag-template & 220 & 0.220 & 0.004 & 0.014 & 0.288 \\
\midrule
\multirow{2}{*}{Yelp}
 & public LLM   &  18 & 0.018 & 0.002 & 0.111 & 0.124 \\
 & tag-template & 749 & 0.749 & 0.000 & 0.000 & 0.239 \\
\bottomrule
\end{tabular}
\end{table}

UnsupportedEmit remains below \(0.01\) across all six cells, while coverage changes substantially.
This illustrates the distinction between discovery and certification: changing the proposer changes which certifiable statements are found without changing the verifier's one-sided honesty property.

\subsection{DP Synthetic Data as a Boundary Case}
\label{app:synthetic-data}

DP synthetic data provides a different reporting contract: a private synthetic corpus \(\widetilde R\) is released first, and subsequent reporting is performed on that proxy rather than directly on \(R\).

We evaluate an Aug-PE-style backend on Yelp using the same private split of \(1800\) records, public statement schema, semantic matcher, and reporting thresholds.
The synthetic pipeline follows the main structure of Private Evolution~\citep{xie2024differentially}, using private voting, Gaussian-noised selection, and LLM-based variation.
The best of four evaluated configurations uses population size \(N=1000\).

The deployable synthetic-data rule emits a statement when
\(
f_s(\widetilde R)\ge\theta.
\)
For evaluation only, we additionally compute its original-corpus support \(f_s(R)\); this quantity is never used by the mechanism.
We report
\[
\mathrm{SupportedRecall}
=
\frac{
|\{s:f_s(R)\ge\theta,\ s\text{ emitted}\}|
}{
|\{s:f_s(R)\ge\theta\}|
}.
\]

\begin{table}[t]
\centering
\small
\caption{\textbf{Matched Yelp comparison with a DP synthetic-data backend.}
UnsupportedEmit for synthetic data is evaluated against support in the original private corpus.}
\label{tab:synthetic-data-backend}
\begin{tabular}{llrrrr}
\toprule
\(\theta\) & Method & EmitRate & UnsupEmit &
Runs w/ any unsup. & SupportedRecall \\
\midrule
0.05 & \textsc{PrivCert}-PF
& 0.823 & 0.000 & \(0/150\) & 0.979 \\
0.05 & Aug-PE-style (\(N=1000\))
& 0.800 & 0.000 & \(0/5\) & 0.952 \\
\midrule
0.10 & \textsc{PrivCert}-PF
& 0.695 & 0.001 & \(2/150\) & 0.938 \\
0.10 & Aug-PE-style (\(N=1000\))
& 0.676 & 0.092 & \(4/5\) & 0.881 \\
\bottomrule
\end{tabular}
\end{table}

At \(\theta=0.10\), the synthetic backend achieves similar coverage but substantially larger original-data UnsupportedEmit.
The repetition units differ---a synthetic-data seed regenerates the complete proxy, whereas a proposal-and-filter repetition resamples final DP noise---so the run-level counts should not be interpreted as matched failure-probability estimates.

The distinction is semantic.
Support in \(\widetilde R\) certifies the proxy, not automatically the original dataset.
An original-data lower certificate would require an additional fidelity guarantee such as
\[
\Pr\left[
\sup_{s\in\mathcal S}
\bigl(
f_s(\widetilde R)-f_s(R)
\bigr)
\le\gamma_+
\right]
\ge
1-\eta.
\]
Under this condition,
\(
f_s(\widetilde R)-\gamma_+\ge\theta
\)
would imply \(f_s(R)\ge\theta\) except with probability at most \(\eta\).
An oracle post-hoc correction based on original-corpus supports removes unsupported emissions in this experiment but reduces EmitRate to approximately \(0.14\)--\(0.28\); because it inspects the private-data oracle, it is not a deployable mechanism.

\subsection{Robustness to the Support Model}
\label{app:matcher-robustness}

The guarantees of \textsc{PrivCert} are defined relative to a fixed support function \(f_s\).
We therefore separate two questions: whether the private certification mechanism remains well calibrated when the support model changes, and whether the identities of supported statements are stable across support models.

\textbf{Alternative NLI matchers.}
We repeat proposal-and-filter with three frozen public NLI matchers: the default DeBERTa cross-encoder, a second DeBERTa model trained on MNLI/FEVER/ANLI, and BART-large-MNLI.
All are fixed independently of the private corpus.

Observed UnsupportedEmit remains below the target \(\beta\) for every matcher, consistent with Theorem~\ref{thm:honesty}.
The emitted content can nevertheless differ substantially.
On Yelp, \(92.7\%\) of the union of statements emitted by any matcher is shared by all three; on WildChat, only \(19.3\%\) is shared by all three and \(39.0\%\) is emitted by only one.
Thus mechanism-level honesty is distinct from semantic agreement between support models.

\textbf{Qwen3-8B as the support model.}
We additionally replace the NLI predicate with a frozen Qwen3-8B judge.
The model and prompt are fixed independently of the private corpus, thinking is disabled, and each record--statement pair is deterministically mapped to a binary support decision.
The resulting empirical-frequency support retains replace-one sensitivity at most \(1/|R|\).

Importantly, we rerun \textsc{PrivCert}-PF using the Qwen-defined support function rather than merely re-scoring outputs produced under the NLI support model.

\begin{table}[t]
\centering
\footnotesize
\caption{\textbf{Robustness of \textsc{PrivCert}-PF to a Qwen3-8B support model.}
Qwen3-8B replaces the default NLI matcher as the operational support predicate.
``In interval'' reports finite-run consistency with the expectation interval of Theorem~\ref{thm:precision-coverage}.}
\label{tab:qwen-support-robustness}
\begin{tabular}{lrrc}
\toprule
Dataset
& UnsupportedEmit
& NLI--Qwen Jaccard
& In interval \\
\midrule
TAB      & 0.0033 & 0.94 & 5/5 \\
WildChat & 0.0023 & 0.33 & 5/5 \\
Yelp     & 0.0037 & 0.78 & 5/5 \\
\bottomrule
\end{tabular}
\end{table}

UnsupportedEmit remains well below \(\beta=0.05\), and all \(15/15\) finite-run EmitRates are empirically consistent with the interval predicted for the expectation by Theorem~\ref{thm:precision-coverage}.
The emitted sets nevertheless differ, especially on WildChat.
Thus the mechanism-level behavior persists under an LLM-based support model, while coverage and certified content remain support-model dependent.

\textbf{Cross-evaluation of frozen free-text outputs.}
The previous experiment changes the support function and reruns the certification mechanism.
We separately test whether the high FalseEmission of free-text baselines in
Table~\ref{tab:false-emission-main} is specific to the default NLI evaluator.

For this diagnostic, released claims are held fixed and only the evaluator is changed from NLI to the same frozen Qwen3-8B judge.
Seven of the fifteen dataset--method cells could be recovered exactly.
Each recovered cell first reproduces its original Table~\ref{tab:false-emission-main}
NLI FalseEmission before Qwen scoring.
No model outputs were regenerated to fill unavailable cells.

\begin{table}[t]
\centering
\footnotesize
\caption{\textbf{Cross-evaluation of exactly recoverable frozen
Table~\ref{tab:false-emission-main} outputs.}
Released claims are fixed and only the operational support evaluator changes from NLI to Qwen3-8B.
Qwen FalseEmission is a descriptive diagnostic and carries no \(\beta\)-guarantee.}
\label{tab:qwen-cross-eval}
\begin{tabular}{llrr}
\toprule
Dataset & Method & NLI FalseEmit & Qwen FalseEmit \\
\midrule
\multirow{2}{*}{TAB}
 & Document-level DP          & 0.894 & 0.470 \\
 & \textsc{PrivCert}-PF       & 0.010 & 0.040 \\
\midrule
\multirow{3}{*}{WildChat}
 & Document-level DP          & 0.729 & 0.857 \\
 & \textsc{InvisibleInk} (Qwen)
                                & 0.619 & 0.929 \\
 & \textsc{PrivCert}-PF       & 0.014 & 0.241 \\
\midrule
\multirow{2}{*}{Yelp}
 & Document-level DP          & 0.803 & 0.239 \\
 & \textsc{PrivCert}-PF       & 0.000 & 0.144 \\
\bottomrule
\end{tabular}
\end{table}

Among the exactly recoverable cells, the ordering between \textsc{PrivCert}-PF and the available free-text baselines is preserved on all three datasets.
The gap remains large on TAB and WildChat but narrows substantially on Yelp.
Thus the qualitative evidence-gap diagnostic is not explained solely by the default NLI matcher, while its magnitude is clearly support-model dependent.

The Qwen scores assigned to the frozen \textsc{PrivCert}-PF outputs do not carry the original \(\beta\)-honesty guarantee because those outputs were certified under the NLI-defined support function.
This cross-evaluation should therefore be interpreted only as a robustness diagnostic.
The remaining Table~\ref{tab:false-emission-main} cells are omitted because  their exact frozen claims were unavailable.

\subsection{Whole-Report Guarantee}
\label{app:family-wise}

The main experiments use the per-candidate guarantee of Theorem~\ref{thm:honesty}.
If instead the complete report must contain no unsupported statement except with probability at most \(\beta_{\mathrm{tot}}\), the union bound in Appendix~\ref{app:proofs-method} allows failure targets satisfying
\(
\sum_{j=1}^m\beta_j
\le
\beta_{\mathrm{tot}}.
\)
With uniform allocation,
\(
\beta_j
=
\frac{\beta_{\mathrm{tot}}}{m}.
\)
For the conservative Laplace calibration used in the main experiments, the
certificate margin becomes
\[
\tau_{\mathrm{FW}}
=
\frac{m\Delta_f}{\varepsilon_{\mathrm{tot}}}
\log\frac{m}{\beta_{\mathrm{tot}}},
\]
which is larger than the per-candidate margin and therefore reduces coverage.

\begin{table}[t]
\centering
\small
\caption{\textbf{Coverage cost of whole-report control.}
The family-wise setting uses
\(\beta_j=\beta_{\mathrm{tot}}/m\) with
\(\beta_{\mathrm{tot}}=0.05\).}
\label{tab:family-wise}
\begin{tabular}{lrr}
\toprule
Dataset &
Per-candidate EmitRate &
Family-wise EmitRate \\
\midrule
TAB      & 0.404 & 0.357 \\
Yelp     & 0.749 & 0.565 \\
WildChat & 0.220 & 0.169 \\
\bottomrule
\end{tabular}
\end{table}

The stronger report-level guarantee reduces coverage on all three datasets.
No unsupported report is observed in these evaluated runs, but the formal statement remains
\[
\Pr[\text{any unsupported emission}]
\le
0.05.
\]
Thus per-candidate calibration provides the higher-coverage operating point used in the main evaluation, while family-wise calibration is available when the guarantee must apply to the complete released report.

\end{document}